%% file: main.tex
\documentclass[submission,copyright,creativecommons]{eptcs}
\providecommand{\event}{GanDALF 2021} % Name of the event you are submitting to
\usepackage{breakurl}             % Not needed if you use pdflatex only.
\usepackage{underscore}           % Only needed if you use pdflatex.

\input{macros}
\usepackage{amsmath,amsthm,amssymb}
\newtheorem{theorem}{Theorem}
\newtheorem{lemma}{Lemma}
\newtheorem{corollary}{Corollary}
\newtheorem{definition}{Definition}

\title{New Algorithms for Combinations of Objectives\\ using Separating Automata}
\author{Ashwani Anand 
\institute{Chennai Mathematical Institute, Chennai, India}
\and
Nathana{\"e}l Fijalkow 
\institute{CNRS, LaBRI, Bordeaux, France}
\institute{The Alan Turing Institute, London, United Kingdom}
\and
Ali{\'e}nor Goubault-Larrecq
\institute{ENS Lyon, Lyon, France}
\and 
J{\'e}r{\^o}me Leroux
\institute{CNRS, LaBRI, Bordeaux, France}
\and
Pierre Ohlmann
\institute{Universit{\'e} de Paris, Paris, France}
}
\def\titlerunning{Combining Objectives using Separating Automata}
\def\authorrunning{A. Anand, N. Fijalkow, A. Goubault-Larrecq, J. Leroux \& P. Ohlmann}
\def\copyrightholders{A. Anand, N. Fijalkow, A. Goubault-Larrecq,\\
                      \mbox{}\hspace{11pt} J. Leroux \& P. Ohlmann}

\begin{document}
\maketitle

\begin{abstract}
The notion of separating automata was introduced by Boja{\'n}czyk and Czerwi{\'n}ski for understanding the first quasipolynomial  time algorithm for parity games. In this paper we show that separating automata is a powerful tool for constructing algorithms solving games with combinations of objectives. We construct two new algorithms: the first for disjunctions of parity and mean payoff objectives, matching the best known complexity, and the second for disjunctions of mean payoff objectives, improving on the state of the art. In both cases the algorithms are obtained through the construction of small separating automata, using as black boxes the existing constructions for parity objectives and for mean payoff objectives.
\end{abstract}

\section{Introduction}
\label{sec:intro}
\input{intro}

\section{Preliminaries}
\label{sec:preliminaries}
\input{preliminaries}

\section{Disjunction of parity and mean payoff}
\label{sec:disjmeanpayoffparity}
\input{disjmeanpayoffparity}

\section{Disjunction of mean payoff}
\label{sec:disjmeanpayoff}
\input{disjmeanpayoff}

\section*{Conclusions}
The conceptual contribution of this paper is to show that the notion of separating automata is a powerful tool for constructing algorithms. We illustrated this point with two applications, constructing algorithms for two classes of objectives.
A key appeal of our approach is that we assumed the existence of separating automata both for parity and for mean payoff objectives and used them as black boxes to construct separating automata for objectives combining them.

It is tempting to use our approach for constructing separating automata for disjunctions of parity objectives.
However, solving games with disjunctions of two parity objectives is NP-complete~\cite{CHP07}, hence we cannot hope for subexponential separating automata.

In this paper we offered upper bounds on the sizes of separating automata for two classes of objectives.
We leave open the questions of proving (matching) lower bounds, which would indicate that this approach cannot
yield better algorithms in these cases.

\bibliographystyle{eptcs}
\bibliography{bib}

\end{document}

%% file: macros.tex
\newcommand{\Z}{\mathbb Z}
\newcommand{\N}{\mathbb N}

\newcommand{\set}[1]{\left\{ #1 \right\}}

\newcommand{\game}{\mathcal{G}}
\newcommand{\auto}{\mathcal{A}}

\newcommand{\VE}{\textrm{V}_{\textrm{Eve}}}
\newcommand{\VA}{\textrm{V}_{\textrm{Adam}}}

\newcommand{\last}{\textrm{last}}

\newcommand{\col}{\textrm{col}}

\newcommand{\MP}[1]{\mathtt{MeanPayoff}_{#1}}

\newcommand{\chain}{\triangleright}

\newcommand{\Path}{\textrm{Path}}
\newcommand{\FPath}{\textrm{Path}_{\textrm{fin}}}
\newcommand{\IPath}{\textrm{Path}_{\infty}}

\newcommand{\Safe}{\mathtt{Safe}}
\newcommand{\Parity}{\mathtt{Parity}}

\newcommand{\len}{\textrm{len}}

\newcommand{\NP}{\textrm{NP}}
\newcommand{\coNP}{\textrm{coNP}}

\usepackage{tikz}
\usetikzlibrary{arrows}
\usetikzlibrary{automata}
\usetikzlibrary{shapes,snakes}
\usetikzlibrary{calc}
\usetikzlibrary{patterns}
\usetikzlibrary{shapes.geometric}
\usetikzlibrary{positioning}
\tikzstyle{every node}=[font=\small]
\tikzstyle{eve}=[circle,minimum size=.3cm,draw=gray!90,inner sep=1pt,fill=gray!20,very thick]
\tikzstyle{adam}=[rounded corners=.5,regular polygon,regular polygon
  sides=4,minimum size=.4cm,draw=gray!90,inner sep=1pt,fill=gray!20,very thick]
\tikzstyle{every edge}=[draw,>=stealth',shorten >=1pt]
\tikzstyle{win}=[fill=green!50,draw=green!70!black]
\tikzstyle{lose}=[fill=white,draw=red!70!black]

\tikzstyle{state}=[draw,circle,minimum size=5mm]
\tikzstyle{accepting}=[double]

%% file: intro.tex
The notion of separating automata was introduced by Boja{\'n}czyk and Czerwi{\'n}ski~\cite{BC18} to give a streamlined presentation of the first quasipolynomial time algorithm for parity games due to Calude, Jain, Khoussainov, Li, and Stephan~\cite{CJKLS17}.
The first observation made by Boja{\'n}czyk and Czerwi{\'n}ski was that the statistics used in that algorithm can be computed by a finite deterministic safety automaton reading a path in the game.
They showed that the property making this particular automaton useful for solving parity games is that, roughly speaking, it separates positional winning paths from losing paths: they called this property separating.
The second observation of Boja{\'n}czyk and Czerwi{\'n}ski was that separating automata yield a natural reduction to safety games, in other words the construction of a separating automaton induces an algorithm whose complexity depends on the size of the automaton.

The notion of separating automata has been further studied in the context of parity games:
Czerwi{\'n}ski, Daviaud, Fijalkow, Jurdzi{\'n}ski, Lazi{\'c}, and Parys~\cite{CDFJLP18} showed that two other quasipolynomial time algorithms can also be presented using the construction of a separating automaton.
The main technical result of~\cite{CDFJLP18} is that separating automata are in some sense equivalent to the notion of universal trees at the heart of the second quasipolynomial time algorithm by Jurdzi{\'n}ski and Lazi{\'c}~\cite{JL17} and formalised by Fijalkow~\cite{Fij18}. The consequence of this equivalence is a quasipolynomial lower bound on the size of separating automata for parity objectives.

\vskip1em
Going beyond parity games, Colcombet and Fijalkow~\cite{CF18,CF19} introduced universal graphs and showed an equivalence result between separating automata and universal graphs for any positionally determined objective.
This paves the way for using separating automata for other classes of objectives. The first work in that direction is due to Fijalkow, Ohlmann, and Gawrychowski~\cite{FGO18}, who obtained matching upper and lower bounds on the size of separating automata for mean payoff objectives, matching the best known (deterministic) complexity for solving mean payoff games.

\vskip1em
The goal of this paper is to show how to construct separating automata for combinations of objectives, thereby obtaining new algorithms for solving the corresponding games.
We will consider two subclasses combining parity and mean payoff objectives, with the following goal: rather than constructing separating automata from scratch, we want to define constructions using separating automata for the atomic objectives as black boxes. In other words, we assume the existence of separating automata for parity objectives (provided in~\cite{CDFJLP18}, see also~\cite{CF18}) and for mean payoff objectives (provided in~\cite{FGO18}), and construct separating automata for combinations of these classes.
An important benefit of this approach is its simplicity: as we will see, both constructions and their correctness proofs are rather short and focus on the interactions between the objectives.

\vskip1em
Section~\ref{sec:preliminaries} introduces separating automata and shows how they yield algorithms by reduction to safety games.
The two classes of objectives we consider are disjunctions of parity and mean payoff objectives in Section~\ref{sec:disjmeanpayoffparity}, and disjunctions of mean payoff objectives in Section~\ref{sec:disjmeanpayoff}.
We refer to the subsections~\ref{subsec:complexitydisjunctionsparitymeanpayoff} and~\ref{subsec:complexitydisjunctionsmeanpayoff}
for related work on solving these games.

%% file: preliminaries.tex
We write $[i,j]$ for the interval $\set{i,i+1,\dots,j-1,j}$, and use parentheses to exclude extremal values,
so for instance $[i,j)$ is $\set{i,i+1,\dots,j-1}$.
We let $C$ denote a set of colours and write $C^*$ for finite sequences of colours (also called finite words),
$C^+$ for finite non-empty sequences, and $C^\omega$ for infinite sequences (also called infinite words).
The empty word is $\varepsilon$.

\subsection{Graphs}

\paragraph*{Graphs}
We consider edge labelled directed graphs: 
a graph $G$ is given by a (finite) set $V$ of vertices and a (finite) set $E \subseteq V \times C \times V$ of edges,
with $C$ a set of colours, so we write $G = (V,E)$.
An edge $(v,c,v')$ is from the vertex $v$ to the vertex $v'$ and is labelled by the colour $c$.
We sometimes refer to $V(G)$ for $V$ and $E(G)$ for $E$ to avoid any ambiguity.
The size of a graph is its number of vertices.
A vertex $v$ for which there exists no outgoing edges $(v,c,v') \in E$ is called a sink.

\paragraph*{Paths}
A path $\pi$ is a (finite or infinite) sequence 
\[
\pi = v_0 c_0 v_1 c_1 v_2 \dots
\]
where for all $i$ we have $(v_i,c_i,v_{i+1}) \in E$.
If it is finite, that is, $\pi= v_0 c_0 \dots c_{i-1} v_i$, we call $\len(\pi) = i \geq 0$ the length of $\pi$, 
and use $\last(\pi)$ to denote the last vertex $v_i$.
The length of an infinite path is $\len(\pi) = \infty$.
We say that $\pi$ starts from $v_0$ or is a path from $v_0$, 
and in the case where $\pi$ is finite we say that $\pi$ is a path ending in $\last(\pi)$ or simply a path to $\last(\pi)$.
We say that $v'$ is reachable from $v$ if there exists a path from $v$ to $v'$.
We let $\pi_{\le i}$ denote the prefix of $\pi$ of length $i$, meaning 
$\pi_{\le i} = v_0 c_0 v_1 \dots c_{i-1} v_i$.
A cycle is a path from a vertex to itself of length at least one.

We use $\FPath(G), \IPath(G)$ and $\Path(G)$ to denote respectively the sets of finite paths of $G$, infinite paths of $G$, and their union. We sometimes drop $G$ when it is clear from context. 
For $v_0 \in V$ we also use $\FPath(G,v_0), \IPath(G,v_0)$ and $\Path(G,v_0)$ 
to refer to the sets of paths starting from $v_0$. 
We use $\col(\pi)$ to denote the (finite or infinite) sequence of colours $c_0 c_1 \dots$ induced by $\pi$.

\paragraph*{Objectives}
An objective is a set $\Omega \subseteq C^\omega$ of infinite sequences of colours.
We say that a sequence of colours belonging to $\Omega$ satisfies $\Omega$, and extend this terminology to infinite paths:
$\pi$ satisfies $\Omega$ if $\col(\pi) \in \Omega$.

\begin{definition}[Graphs satisfying an objective]
Let $\Omega$ be an objective and $G$ a graph.
We say that $G$ satisfies $\Omega$ if all infinite paths in $G$ satisfy $\Omega$.
\end{definition}

\paragraph*{Safety automata}
A deterministic safety automaton over the alphabet $C$ is given by a finite set of states $Q$, an initial state $q_0 \in Q$, 
and a transition function $\delta : Q \times C \to Q$, so we write $\auto = (Q,q_0, \delta)$.
Note that $\delta$ is a partial function, meaning that it may be that $\delta(q,c)$ is undefined for some $q,c \in Q \times C$.
Such an automaton induces a graph whose set of vertices is $Q$ and set of edges is 
$E = \set{(q, \delta(q,c)) : q \in Q, c \in C}$; we therefore use the terminology for graphs
to speak about automata, and identify an automaton and the graph it induces.
Since we are only considering deterministic safety automata in this paper, 
we omit the adjectives deterministic and safety and simply speak of an automaton.
We extend $\delta$ to sequences of colours by the formulas 
$\delta^*(q,\epsilon) = q$ and $\delta^*(q,wc) = \delta(\delta^*(q,w),c)$.
The language recognised by an automaton $\auto$ is $L(\auto)$ defined by
\[
L(\auto) = \set{\col(\pi) : \pi \text{ infinite path from $q_0$}}.
\]
Note that if $w$ is a finite prefix of a word in $L(\auto)$, then $\delta^*(q_0,w)$ is well defined.

\subsection{Games}

\paragraph*{Arenas}
An arena is given by a graph $G$ together with a partition $V = \VE \uplus \VA$ of its set of vertices describing which player controls each vertex.
%We represent vertices controlled by Eve with circles and those controlled by Adam with squares.

\paragraph*{Games}
A game is given by an arena and an objective $\Omega$.
We often let $\game$ denote a game, its size is the size of the underlying graph.
It is played as follows.
A token is initially placed on some vertex $v_0$, and the player who controls this vertex pushes the token along an edge, 
reaching a new vertex; the player who controls this new vertex takes over
and this interaction goes on either forever and describing an infinite path or until reaching a sink.

We say that a path is winning\footnote{We always take the point of view of Eve, so winning means winning for Eve, and similarly a strategy is a strategy for Eve.} if it is infinite and satisfies $\Omega$, or finite and ends in a sink.
The definition of a winning path includes the following usual convention: if a player cannot move they lose,
in other words sinks controlled by Adam are winning (for Eve) and sinks controlled by Eve are losing.

We extend the notations $\FPath, \IPath$ and $\Path$ to games by considering the underlying graph.

\paragraph*{Strategies}
A strategy in $\game$ is a partial map $\sigma : \FPath(\game) \to E$ such that $\sigma(\pi)$ is 
an outgoing edge of $\last(\pi)$ when it is defined.
We say that a path $\pi = v_0 c_0 v_1 \dots$ is consistent with $\sigma$ if
for all $i < \len(\pi)$, if $v_i \in \VE$ then $\sigma$ is defined over $\pi_{\leq i}$ and 
$\sigma(\pi_{\le i}) = (v_i, c_i, v_{i+1})$.
A consistent path with $\sigma$ is maximal if it is not the strict prefix of a consistent path with $\sigma$
(in particular infinite consistent paths are maximal).

A strategy $\sigma$ is winning from $v_0$ if all maximal paths consistent with $\sigma$ are winning.
Note that in particular, if a finite path $\pi$ is consistent with a winning strategy $\sigma$ and ends in a vertex which belongs to Eve, then $\sigma$ is defined over $\pi$.
We say that $v_0$ is a winning vertex of $\game$ or that Eve wins from $v$ in $\game$ 
if there exists a winning strategy from $v_0$.
% and let $\WE(\game)$ denote the set of winning vertices.

\paragraph*{Positional strategies}
Positional strategies make decisions only considering the current vertex. 
Such a strategy is given by $\widehat \sigma : \VE \to E$.
A positional strategy induces a strategy $\sigma : \FPath \to E$ from any vertex $v_0$ 
by setting ${\sigma}(\pi) = \widehat{\sigma}(\last(\pi))$ when $\last(\pi) \in \VE$.

\begin{definition}
We say that an objective $\Omega$ is \textit{positionally determined} if for every game with objective $\Omega$ and vertex $v_0$,
if Eve wins from $v_0$ then there exists a positional winning strategy from $v_0$.
\end{definition}

Given a game $\game$, a vertex $v_0$, and a positional strategy $\sigma$ we let $\game[\sigma,v_0]$ denote the graph obtained by restricting $\game$ to vertices reachable from $v_0$ by playing $\sigma$ and to the moves prescribed by $\sigma$.
Formally, the set of vertices and edges is
\[
\begin{array}{lll}
V[\sigma,v_0] & = & \set{v \in V : \text{there exists a path from } v_0 \text{ to } v \text{ consistent with } \sigma}, \\
E[\sigma,v_0] & = & \set{(v,c,v') \in E : v \in \VA \text{ or } \left( v \in \VE \text{ and }\right.
\left. \sigma(v) = (v,c,v') \right)} \\
& \cap & V[\sigma,v_0] \times C \times V[\sigma,v_0].
\end{array}
\]

In this paper we will consider prefix independent objectives for technical convenience.

\begin{definition}
We say that an objective $\Omega$ is prefix independent if for all $u \in C^*$ and $v \in C^\omega$, 
we have $uv \in \Omega$ if and only if $v \in \Omega$.
\end{definition}

\begin{lemma}
Let $\Omega$ be a prefix independent objective, $\game$ a game, $v_0$ a vertex, and $\sigma$ a positional strategy.
Then $\sigma$ is winning from $v_0$ if and only if the graph $\game[\sigma,v_0]$ satisfies $\Omega$
and does not contain any sink controlled by Eve.
\end{lemma}

\subsection*{Solving games}
\paragraph*{Decision problem}
The decision problem we consider in this paper, called solving a game, is the following:
given a game $\game$ and an initial vertex $v_0$, does Eve have a winning strategy from $v_0$ in $\game$?
Each objective yields a class of games, so we speak for instance of ``solving mean payoff games''.

\paragraph*{Computational model}
The complexity of solving a game depends on a number of parameters originating either from the underlying graph or the objective.
Some of the objectives involve rational numbers.
The typical parameters from the underlying graph are the number $n$ of vertices and the number $m$ of edges.

We use the classical Random Access Model (RAM) of computation with fixed word size,
which is the size of a memory cell on which arithmetic operations take constant time.
We specify for each class of objectives the word size.

\subsection*{Reduction to safety games}
\paragraph*{Safety games}
The safety objective $\Safe$ is defined over the set of colours $C = \set{\varepsilon}$ by $\Safe = \set{\varepsilon^\omega}$: 
in words, all infinite paths are winning, so losing for Eve can only result from reaching a sink that she controls.
Since there is a unique colour, when manipulating safety games we ignore the colour component for edges.

Note that in safety games, strategies can equivalently be defined as maps to $V$ (and not $E$):
only the target of an edge matters when the source is fixed, since there is a unique colour.
We use this abuse of notations for the sake of simplicity and conciseness.

\begin{theorem}\label{thm:safetygames}
There exists an algorithm in the RAM model with word size $w = \log(n)$ computing the set of winning vertices of a safety game running in time $O(m)$ and space $O(n)$.
\end{theorem}

\paragraph*{Reduction using safety automata}
Let $\auto = (Q, q_0, \delta)$ be an automaton and $\game$ a game with objective $\Omega$.
We define the chained game $\game \chain \auto$ as the safety game with vertices $V' = \VE' \uplus \VA'$ and edges $E'$ given by
\[
\begin{array}{lll}
\VE' & = & \left(\VE \times Q\right)\ \cup \set{\bot}, \\
\VA' & = & \VA \times Q, \\
E' & = & \{((v,q), \varepsilon, (v',\delta(q,c))) : q \in Q, (v,c,v') \in E, \delta(q,c) \text{ is defined}\} \\
 & \cup & \{((v,q), \varepsilon, \bot) : q \in Q, (v,c,v') \in E, \delta(q,c) \text{ is not defined}\}.
\end{array}
\]

In words, from $(v,q) \in V \times Q$, the player whom $v$ belongs to chooses an edge $(v,c,v') \in E$, and the game progresses to $(v',\delta(q,c))$ if $q$ has an outgoing edge with colour $c$ in $\auto$, and to $\bot$ otherwise, which is losing for Eve.

Note that the obtained game $\game \chain \auto$ is a safety game: Eve wins if she can play forever or end up in a sink controlled by Adam.

\begin{lemma}\label{lem:productharder}
Let $\Omega$ be an objective, $\game$ a game with objective $\Omega$, $\auto$ an automaton, and $v_0 \in V$.
If $\auto$ satisfies $\Omega$ and Eve wins from $(v_0,q_0)$ in $\game \chain \auto$, 
then Eve wins from $v_0$ in $\game$.
\end{lemma}

\begin{proof}
Let $\sigma'$ be a winning strategy from $(v_0, q_0)$ in $\game \chain \auto$.
We construct a strategy $\sigma$ from $v_0$ in $\game$ which simulates $\sigma'$ in the following sense,
which we call the simulation property:
\begin{center}
for any path $\pi = v_0 c_0 \dots c_{i-1} v_i$ in $\game$ consistent with $\sigma$, \\
there exists a path $\pi' = (v_0, q_0) \dots (v_i, q_i)$ in $\game \chain \auto$ consistent with $\sigma'$.
\end{center}

We define $\sigma$ over finite paths $v_0c_0 \dots c_{i-1}v_i$ with $v_i \in \VE$ by induction over $i$, so that the simulation property holds.
For $i = 0$, $\pi'=(v_0,q_0)$ is a path in $\game \chain \auto$ consistent with $\sigma'$.
Let $\pi = v_0 c_0 \dots c_{i-1} v_{i}$ be a path in $\game$ consistent with $\sigma$ and $v_i \in \VE$, 
we want to define $\sigma(\pi)$.
Thanks to the simulation property there exists a path 
$\pi' = (v_0, q_0) \dots (v_{i}, q_{i})$ in $\game \chain \auto$ consistent with $\sigma'$.
Then $(v_i,q_i) \in \VE'$, and since $\sigma'$ is winning it is defined over $\pi'$, 
let us write $\sigma'(\pi') = (v_{i+1}, \delta(q_i, c_i))$ with $(v_i,c_i,v_{i+1}) \in E$.
We set $\sigma(\pi) = v_{i+1}$.
%This defines $\sigma$ over consistent paths of length at most $i$ ending in $\VE$.

To conclude the definition of $\sigma$ we need to show that the simulation property extends to paths
of length $i+1$.
Let $\pi_{i+1} = \pi\ c_i v_{i+1} = v_0 c_0 \dots v_i c_i v_{i+1}$ be consistent with $\sigma$.
We apply the simulation property to $\pi$ to construct 
$\pi' = (v_0, q_0) \dots (v_{i}, q_{i})$ a path in $\game \chain \auto$ consistent with $\sigma'$.
Let us consider $\pi'_{+1} = \pi'\ (v_{i+1}, \delta(q_i, c_i))$ with $(v_i,c_i,v_{i+1}) \in E$, 
we claim that $\pi'_{+1}$ is consistent with $\sigma'$.
Indeed, if $v_i \in \VA$, there is nothing to prove, and if $v_i \in \VE$, this holds by construction: 
since $\sigma(\pi) = v_{i+1}$ we have $\sigma'(\pi') = (v_{i+1}, \delta(q_i, c_i))$.
This concludes the inductive proof of the simulation property together with the definition of $\sigma$.

\vskip1em
We now prove that $\sigma$ is a winning strategy from $v_0$. 
Let $\pi=v_0 c_0 v_1 \dots$ be a maximal consistent path with $\sigma$, and let $\pi' = (v_0,q_0)(v_1,q_1) \dots$ be the corresponding path in $\game \chain \auto$ consistent with $\sigma'$.
Let us first assume that $\pi$ is finite, and let $v_i = \last(\pi)$. If $v_i \in \VE$, then by construction $\sigma$ is defined over $\pi$, so the path $v_0c_0 \dots v_i c_i v_{i+1}$, where $\sigma(\pi)=v_{i+1}$ is consistent with $\sigma$, contradicting maximality of $\pi$. If however $v_i \in \VA$, then by maximality of $\sigma$, $v_i$ must be a sink, hence $\pi$ is winning.
%If $\pi$ is finite then it ends in a sink $v_i$, then so is $(v_i,q_i)$, hence $\pi'$ is finite and maximal.
%Since $\sigma'$ is winning, $\pi'$ is winning, hence $(v_i,q_i)$ belongs to Adam in $\game \chain \auto$, so $v_i$ belongs to Adam in $\game$ hence $\pi$ is winning.
Now if $\pi$ is infinite then so is $\pi'$, and then $q_0 c_0 q_1 c_1 \dots$ is a path in $\auto$, so $\col(\pi) = c_0 c_1 \dots \in L(\auto) \subseteq \Omega$, and $\pi$ is winning.
We conclude that $\sigma$ is a winning strategy from $v_0$ in $\game$.
\end{proof}

It is not hard to see that if $\auto$ is deterministic and recognises exactly $\Omega$, then $\game$ and $\game \chain \auto$ are equivalent.
However, for many objectives $\Omega$ (for instance, parity or mean payoff objectives), a simple topological argument shows that such deterministic automata do not exist.
Separating automata are defined by introducing $n$ as a parameter, and relaxing the condition $L(\auto)=\Omega$ to the weaker condition $\Omega^{\mid n} \subseteq L(\auto) \subseteq \Omega$, where $\Omega^{\mid n}$ is the set of infinite sequence of colours that label paths from graphs of size at most $n$ satisfying $\Omega$.
Formally\footnote{Here we use the fact that $\Omega$ is prefix independent to simplify the definition: we are considering infinite paths from any vertex of the graph. To extend this definition beyond prefix independent objectives we would need to fix an initial vertex.},
\[
\Omega^{\mid n} = \{\col(\pi) \mid \pi \in \IPath(G), G \text{ has size at most $n$ and satisfies $\Omega$}\}.
\]

\begin{definition}
An $(n, \Omega)$-separating automaton $\auto$ is an automaton such that
\[
\Omega^{\mid n} \subseteq L(\auto) \subseteq \Omega.
\]
\end{definition}

The definition given here differs from the original one given by Boja{\'n}czyk and Czerwi{\'n}ski~\cite{BC18},
who use a different relaxation $\Omega_{\mid n}$ satisfying $\Omega^{\mid n} \subseteq \Omega_{\mid n} \subseteq \Omega$.
Therefore a separating automaton in the sense of~\cite{BC18} is also a separating automaton in our sense (but not conversely). 

\begin{theorem}
\label{thm:separatingautomata}
Let $\Omega$ be a positionally determined objective, $\auto$ an $(n,\Omega)$-separating automaton, $\game$ a game of size $n$, and $v_0 \in V$. 
Then Eve wins from $v_0$ in $\game$ if and only if she wins from $(v_0,q_0)$ in $\game \chain \auto$.
\end{theorem}

\begin{proof}
The ``if'' directly follows from Lemma~\ref{lem:productharder}. 
Conversely, assume that Eve wins from $v_0$ in $\game$, and let $\sigma$ be a strategy from $v_0$ in $\game$, 
which we choose to be positional. 
As explained in the definition of safety games, without loss of generality we can see $\sigma$ as a function $\sigma: \VE \to V$.
We define a positional strategy $\sigma' : \VE' \to V'$ by
\[
\sigma'(v,q) = (\sigma(v),\delta(q,c)) \text{ with } (v,c,v') \in E.
\]
Let $\pi' = (v_0,q_0) \dots (v_i,q_i)$, with for all $j \leq i, e_j = (v_j,c_j,v_{j+1}) \in E$, 
be a finite path in $\game \chain \auto$ consistent with $\sigma'$, and let $\pi = v_0 c_0 \dots c_i v_{i+1}$.
Then by definition of $\sigma'$, $\pi$ is consistent with $\sigma$, which rephrases a being a path in $\game[\sigma, v_0]$.
Now, $\game[\sigma, v_0]$ is a graph satisfying $\Omega$, so $\col(\pi)=c_0 \dots c_i$ is a prefix of a word in $\Omega^{\mid n} \subseteq L(\auto)$.
This implies that $\delta^*(q_0,c_0 \dots c_i)$ is well defined, hence so is $\delta(q_i, c_i)$.
In particular, $\sigma'$ induces a strategy in $\game \chain \auto$ from $(v_0, q_0)$. We now prove that it is winning.

Let $\pi' = (v_0,q_0) \dots$ be a maximal path in $\game$ consistent with $\sigma'$, which we assume to be finite for contradiction.
Let $\last(\pi') = (v_i,q_i)$ be a sink, and define $\pi = v_0 c_0 \dots c_{i-1} v_i$.
By definition of $\pi'$, it holds that $\pi$ is a path in $\game$ which is consistent with $\sigma$.
Since $\pi$ is finite and $\sigma$ is winning $v_i$ cannot be a sink, it has an outgoing edge $e_i=(v_i,c,v') \in E$.
Then $((v_i,q_i), \varepsilon, (v_{i+1}, \delta(q_i, c_i)) \in E'$, so $(v_i,q_i)$ is not a sink: a contradiction.
Hence $\pi'$ is an infinite path in the safety game $\game \chain \auto$; it is winning by definition.
\end{proof}

\subsection*{Existing constructions for parity and mean payoff games}
\paragraph*{Parity games}
The parity objective is defined over the set of colours $[0,d] \subseteq \N$, as follows:
\[
\Parity_d = \set{w \in [0,d]^\omega : \text{the largest priority appearing infinitely many times in $w$ is even}}.
\]

As mentioned above, the definition of separating automata given in~\cite{BC18} is slightly different, but the result below indeed holds for both definitions, as explained in~\cite{CDFJLP18}.

\begin{theorem}[\cite{BC18,CDFJLP18}]
\label{thm:separatingautomataparity}
Let $n,d \in \N$.
There exists an $(n,\Parity_d)$-separating automaton of size 
\[
O\left(n \cdot \binom{\lceil \log(n) \rceil + d/2 - 1}{\lceil \log(n) \rceil}\right).
\]
\end{theorem}

\paragraph*{Mean payoff games}
The mean payoff objective is defined over the set of colours $[-N,N] \subseteq \Z$, as follows:
\[
\MP{N} = \set{w \in [-N,N]^\omega : \liminf_n \frac{1}{n} \sum_{i = 1}^n w_i \ge 0}.
\]

\begin{theorem}[\cite{FGO18}]
\label{thm:separatingautomatamp}
Let $n,N \in \N$. There exists an $(n,\MP{N})$-separating automaton of size $O(nN)$.
\end{theorem}

%% file: disjmeanpayoffparity.tex
We define the objective $\Parity_d \vee \MP{N}$ referred to as `disjunction of parity and mean payoff'.
The set of colours is $[0,d] \times [-N,N]$.
For $w \in ([0,d] \times [-N,N])^\omega$ we write $w_P \in [0,d]^\omega$ for the projection on the first component
and $w_{MP} \in [-N,N]^\omega$ for the projection on the second component.
\[
\Parity_d \vee \MP{N} = \set{w \in ([0,d] \times [-N,N])^\omega : w_P \in \Parity_d \vee w_{MP} \in \MP{N}}.
\]
We refer to Subsection~\ref{subsec:complexitydisjunctionsparitymeanpayoff} 
for a discussion on existing results.

\begin{theorem}[\cite{ChatterjeeHJ05}]
Disjunctions of parity and mean payoff objectives are prefix independent and positionally determined.
\end{theorem}

\subsection{Separating automata for disjunctions of parity and mean payoff objectives} 
\label{subsec:universalgraphsdisjunctionsparitymeanpayoff}

\begin{theorem}
\label{thm:separatingautomatadisjparitymp}
Let $n,d,N \in \N$.
Let $\auto_P$ an $(n,\Parity_d)$-separating automaton, $\auto_{MP}$ an $(n,\MP{N})$-separating automaton.
Then there exists an $(n, \Parity_d \vee \MP{N})$-separating automaton of size $O(d \cdot |\auto_P| \cdot |\auto_{MP}|)$.
\end{theorem}

\begin{proof}
Let us write $\auto_P = (Q_P,q_{0,P},\delta_P)$ and $\auto_{MP} = (Q_{MP},q_{0,MP},\delta_{MP})$.
We define a deterministic automaton $\auto_{P \vee MP}$: 
the set of states is $[0,d] \times Q_P \times Q_{MP}$,
the initial state is $(d, q_{0,P}, q_{0,MP})$,
and the transition function is
\[
\delta((p,q_P,q_{MP}), (p',w)) = 
\begin{cases}
(\max(p,p'), q_P, \delta_{MP}(q_{MP}, w)) & \text{ if $\delta_{MP}(q_{MP}, w)$ is defined,} \\
(0, \delta_P(q_{P}, \max(p,p')), q_{0,MP}) & \text{ if $\delta_{MP}(q,_{MP}, w)$ is not defined}.
\end{cases}
\]
Intuitively: $\auto_{P \vee MP}$ simulates the automaton $\auto_{MP}$, storing the maximal priority seen since the last reset (or from the beginning).
If the automaton $\auto_{MP}$ rejects, the automaton resets, which means two things:
it simulates one transition of the automaton $\auto_{P}$ using the stored priority, 
and resets the state of $\auto_{MP}$ to its initial state.
The automaton $\auto_{P \vee MP}$ rejects only if $\auto_{P}$ rejects during a reset.

We now prove that $\auto_{P \vee MP}$ is an $(n,\Parity_d \vee \MP{N})$-separating automaton.
\begin{itemize}
	\item $L(\auto_{P \vee MP}) \subseteq \Parity_d \vee \MP{N}$.
	Let $\pi$ be an infinite path accepted by $\auto_{P \vee MP}$, we distinguish two cases by looking at the run of $\pi$.
	We extend the previous notation for projections: 
	for a path $\pi$, we write $\pi_P$ for its projection on the first component
	and $\pi_{MP}$ for its projection on the second component.

	\begin{itemize}
		\item If the run is reset finitely many times, let us write $\pi = \pi' \pi''$
		where $\pi'$ is finite and $\pi'_{MP}$ rejected by $\auto_{MP}$, and $\pi''$ is infinite and $\pi''_{MP}$ is accepted by $\auto_{MP}$.
		Since $\auto_{MP}$ satisfies $\MP{N}$, this implies that $\pi''_{MP}$ satisfies $\MP{N}$,
		so by prefix independence, $\pi_{MP}$ satisfies $\MP{N}$ hence $\pi$ satisfies $\Parity_d \vee \MP{N}$.
	
		\item If the run is reset infinitely many times, we may find an infinite decomposition $\pi = \pi^1 \pi^2 \dots$
		where for each $i \in \N$, the path $\pi^i_{MP}$ is rejected by $\auto_{MP}$, and its proper prefixes are not.
		Let $p_i$ be the maximum priority appearing in $\pi^i_{P}$, the run of $\pi$ over $\auto_{P \vee MP}$
		induces a run of $p_1 p_2 \dots$ over $\auto_P$.
		Since $\auto_P$ satisfies $\Parity_d$, this implies that $p_1 p_2 \dots$ satisfies $\Parity_d$.
		By definition of the $p_i$'s, this implies that $\pi_P$ satisfies $\Parity_d$ so a fortiori $\pi$ satisfies $\Parity_d \vee \MP{N}$.
	\end{itemize}		
	
	\item $(\Parity_d \vee \MP{N})^{\mid n} \subseteq L(\auto_{P \vee MP})$.
	Let $G$ be a graph satisfying $\Parity_d \vee \MP{N}$ of size at most $n$,
	we need to show that all infinite paths of $G$ are accepted by $\auto_{P \vee MP}$.
	
	We construct a graph $G_{MP}$ over the set of colours $[-N,N]$ and a graph $G_P$ over the set of colours $[0,d]$.
	Both use the same set of vertices as $G$. 
	We prove that the graph $G_{MP}$ satisfies $\MP{N}$, which is used to prove that $G_P$ satisfies $\Parity_d$.
	\begin{itemize}
		\item \textbf{The graph $G_{MP}$.} 
	There is an edge $(v,w,v') \in E(G_{MP})$ if 
	there exists $p \in [0,d]$ such that $e = (v,(p,w),v') \in E(G)$,
	and either $p$ is odd or $p$ is even and $e$ is not contained in any negative cycle with maximum priority $p$.
	
	\vskip1em	
	We claim that $G_{MP}$ satisfies $\MP{N}$.
	Assume towards contradiction that $G_{MP}$ contains a negative cycle $C_{MP}$. 
	It induces a negative cycle $C$ in $G$, by definition of $G_{MP}$ necessarily the maximum priority in $C$ is odd.
	Hence $C$ is a negative odd cycle in $G$, a contradiction.
	
		\item \textbf{The graph $G_{P}$.}
	There is an edge $(v,p,v') \in E(G_P)$ if 
	there exists a path in $G$ from $v$ to $v'$ 
	with maximum priority $p$ and (whose projection on the mean payoff component is) rejected by $\auto_{MP}$.

	\vskip1em	
	We claim that $G_{P}$ satisfies $\Parity_d$.
	Assume towards contradiction that $G_P$ contains an odd cycle $C_P$.
	For each edge in this cycle there is a corresponding path rejected by $\auto_{MP}$
	with the same maximum priority.
	Putting these paths together yields an odd cycle $C$, of maximal priority $p$, in $G$ whose projection on the mean payoff component 
	is rejected by $\auto_{MP}$.
	Since $\MP{N}^{\mid n} \subseteq L(\auto_{MP})$ and as we have shown, $G_{MP}$ satisfies $\MP{N}$, 
	the projection of $C$ on the mean payoff component is not in $G_{MP}$,
	so there exists an edge $(v,(p',w),v')$ in $C$ such that $(v,w,v')$ is not in $E(G_{MP})$. 
	This implies that $p'$ is even, so in particular $p' < p$, and $(v,(p',w),v')$ is contained in a negative cycle $C'$ in $G$ with maximum priority $p'$.	
	Combining the odd cycle $C$ followed by sufficiently many iterations of the negative cycle $C'$ yields a path in $G$, with negative weight, and maximal priority $p$ which is odd, a contradiction.
	\end{itemize}
			
	Let $\pi$ an infinite path in $G$, we show that $\pi$ is accepted by $\auto_{P \vee MP}$.
	Let us consider first only the mean payoff component: we run $\pi$ repeatedly 	over $\auto_{MP}$,
	and distinguish two cases.
	
	\begin{itemize}
		\item If there are finitely many resets, let us write $\pi = \pi_1 \pi_2 \dots \pi_k \pi'$
		where $\pi_1,\dots,\pi_k$ are paths rejected by $\auto_{MP}$ with proper prefixes accepted by $\auto_{MP}$ and $\pi'$ is accepted by $\auto_{MP}$.
		To show that $\pi$ is accepted by $\auto_{P \vee MP}$ we need to show that the automaton $\auto_P$
		accepts the word $p_1 \dots p_k$ where $p_i$ is the maximum priority appearing in $\pi_i$ for $i \in [1,k]$.
		Indeed, $p_1 \dots p_k$ is a path in $G_P$, which is a graph of $n$ satisfying $\Parity_d$,
		so $\auto_P$ accepts $p_1 \dots p_k$.
	
		\item If there are infinitely many resets, let us write $\pi = \pi_1 \pi_2 \dots$
		where for each $i \in \N$, the path $\pi_i$ is rejected by $\auto_{MP}$ and its proper prefixes are not.
		To show that $\pi$ is accepted by $\auto_{P \vee MP}$ we need to show that the automaton $\auto_P$
		accepts the word $p_1 p_2 \dots$, which holds for the same reason as the other case: 
		$p_1 p_2 \dots$ is a path in $G_P$.
	\end{itemize}		
\end{itemize}

\end{proof}

\subsection{The complexity of solving disjunctions of parity and mean payoff games using separating automata}
\label{subsec:complexitydisjunctionsparitymeanpayoff}

The class of games with a disjunction of a parity and a mean payoff objective have been introduced in~\cite{ChatterjeeHJ05},
with a twist: this paper studies the case of a conjunction instead of a disjunction, which is more natural in many applications.
This is equivalent here since both parity and mean payoff objectives are dual, in other words
if the objective of Eve is a disjunction of a parity and a mean payoff objective, then 
the objective of Adam is a conjunction of a parity and a mean payoff objective.
Hence all existing results apply here with suitable changes.
The reason why we consider the objective of the opponent is that it is positionally determined,
which is the key assumption for using the separating automaton technology.

The state of the art for solving disjunctions of parity and mean payoff games is due to~\cite{DaviaudJL18},
which presents a pseudo-quasi-polynomial algorithm.
We refer to~\cite{DaviaudJL18} for references on previous studies for this class of games.
As they explain, these games are logarithmic space equivalent to the same games replacing mean payoff by energy,
and polynomial time equivalent to games with weights~\cite{ScheweWZ19}, extending games with costs~\cite{FM14}.

Combining Theorem~\ref{thm:separatingautomatadisjparitymp}, Theorem~\ref{thm:separatingautomataparity}, Theorem~\ref{thm:separatingautomatamp}, and Theorem~\ref{thm:safetygames} yields the following result.

\begin{theorem}
Let $n,d,N \in \N$.
There exists an algorithm in the RAM model with word size $w = \log(n) + \log(N)$
for solving disjunctions of parity and mean payoff games
with $n$ vertices, $m$ edges, weights in $[-N,N]$ and priorities in $[0,d]$
of time complexity 
\[
O\left(m d \cdot \underbrace{n \cdot \binom{\lceil \log(n) \rceil + d/2 - 1}{\lceil \log(n) \rceil}}_{\text{Parity}}
\cdot \underbrace{nN}_{\text{Mean Payoff}}
\right).
\]
and space complexity $O(n)$.
\end{theorem}

Our algorithm is similar to the one constructed in~\cite{DaviaudJL18}: they are both value iteration algorithms (called progress measure lifting algorithm in~\cite{DaviaudJL18}), combining the two value iteration algorithms for parity and mean payoff games.
However, the set of values are not the same (our algorithm stores an additional priority) and the proofs are very different. 
Besides being much shorter, one advantage of our proof is that it works with abstract separating automata for both parity and mean payoff objectives, and shows how to combine them, whereas in~\cite{DaviaudJL18} the proof is done from scratch, extending both proofs for the parity and mean payoff objectives.

%% file: disjmeanpayoff.tex
We define the objective $\bigvee_{i \in [1,d]} \MP{N}^i$ referred to as `disjunction of mean payoff'.
The set of colours is $[-N,N]^d \subseteq \Z^d$.
For $w \in ([-N,N]^d)^\omega$ and $i \in [1,d]$ we write $w_i \in [-N,N]^\omega$ for the projection on the $i$-th component.
\[
\bigvee_{i \in [1,d]} \MP{N}^i = \set{w \in ([-N,N]^d)^\omega : \exists i \in [1,d], w_i \in \MP{N}}.
\]
We refer to Subsection~\ref{subsec:complexitydisjunctionsmeanpayoff} for a discussion on existing results.

\begin{theorem}[Follows from~\cite{Kopczynski06} as observed in~\cite{VelnerC0HRR15}]
Disjunctions of mean payoff objectives are prefix independent and positionally determined.
\end{theorem}

\subsection{A general reduction of separating automata for strongly connected graphs} 
\label{subsec:reductionstronglyconnected}
The first idea is very general, it roughly says that if we know how to construct separating automata for strongly connected graphs,
then we can use them to construct separating automata for general graphs.

We say that a graph is strongly connected if for every pair of vertices there exists a path from one vertex to the other.
Let us refine the notion of separating automata.
We write $\Omega^{\mid n}_{\text{sc}}$ for the set of infinite sequences of colours that label paths from strongly connected graphs of size at most $n$ satisfying $\Omega$.
Formally,
\[
\Omega^{\mid n}_{\text{sc}} = \{\col(\pi) \mid \pi \in \IPath(G), G \text{ is strongly connected, has size at most $n$, and satisfies $\Omega$}\}.
\]

\begin{definition}[Separating automata for strongly connected graphs]
An automaton $\auto$ is $(n,\Omega)$-separating for strongly connected graphs if 
$\Omega^{\mid n}_{\text{sc}} \subseteq L(\auto) \subseteq \Omega$.
\end{definition}

Let us give a first construction, which we refine later.
Let $\auto_1 = (Q_1,q_{0,1},\delta_1)$ and $\auto_2 = (Q_2,q_{0,2},\delta_2)$ two safety automata, we define their sequential product $\langle \auto_1,\auto_2 \rangle$ as follows:
the set of states is $Q_1 \cup Q_2$, the initial state is $q_{0,1}$, and the transition function is
\[
\delta(s,c) = 
\begin{cases}
\delta_1(s,c) & \text{ if $s \in Q_1$ and $\delta_1(s,c)$ is defined}, \\
q_{0,2}       & \text{ if $s \in Q_1$ and $\delta_1(s,c)$ is not defined}, \\
\delta_2(s,c) & \text{ if $s \in Q_2$ and $\delta_2(s,c)$ is defined}.
\end{cases}
\]
The sequential product is extended inductively:
$\langle \auto_1,\dots,\auto_p \rangle = \langle \langle \auto_1,\dots,\auto_{p-1} \rangle, \auto_p \rangle$.

\begin{lemma}
\label{lem:stronglyconnectedtogeneralnaive}
Let $\Omega$ be a prefix independent objective, $n \in \N$, 
and $\auto$ an $(n,\Omega)$-separating automaton for strongly connected graphs.
Then $\auto^n = \langle \underbrace{\auto,\dots,\auto}_{n \text{ times}} \rangle$ is an $(n,\Omega)$-separating automaton.
\end{lemma}

\begin{proof}
We first show that $L(\auto^n) \subseteq \Omega$.
We note that any infinite run in $\auto^n$ eventually remains in one copy of $\auto$. Since $\auto$ satisfies $\Omega$ and by prefix independence of $\Omega$, this implies that the infinite path satisfies $\Omega$, thus so does $\auto^n$.

Let $G$ be a graph satisfying $\Omega$, we show that $\IPath(G) \subseteq L(\auto^n)$.
We decompose $G$ into strongly connected components: 
let $G_1,\dots,G_p$ be the maximal strongly connected components in $G$
indexed such that if there exists an edge from $G_i$ to $G_j$ then $i < j$.
Then $V(G)$ is the disjoint union of the $V(G_i)$'s.
Each $G_i$ is a subgraph of $G$, and since $G$ satisfies $\Omega$ then so does $G_i$.
It follows that for each $i$ we have $\IPath(G_i) \subseteq L(\auto)$.

Let us consider a path $\pi$ in $G$, we show that $\pi$ is accepted by $\auto^n$. 
The proof is by induction on the number of strongly connected components that $\pi$ traverses.
If it traverses only one such component, say $G_i$, since $\IPath(G_i) \subseteq L(\auto)$
then indeed $\auto$ accepts $\pi$, so a fortiori $\auto^n$ accepts $\pi$.
Otherwise, let $G_i$ the first strongly connected component traversed by $\pi$.
Since $\IPath(G_i) \subseteq L(\auto)$, the run on $\pi$ remains in the first copy of $\auto$ at least as long as $\pi$ remains in $G_i$.
If the run remains in the first copy of $\auto$ forever, then $\pi$ is accepted by $\auto^n$.
If not, the run jumps to the second component to read a suffix of $\pi$,
which traverses one less strongly connected component, so by induction hypothesis 
this suffix is accepted by $\auto^{n-1}$, hence $\pi$ is accepted by $\auto^n$.
\end{proof}

In the construction above we have used the fact that $G$ decomposes into at most $n$ strongly connected components of size $n$.
We refine this argument: the total size of the strongly connected components is $n$.
The issue we are facing in taking advantage of this observation is that the sequence of sizes is not known a priori, it depends on the graph. 
To address this we use a universal sequence. 
First, here a sequence means a finite sequence of non-negative integer, for instance $(5,2,3,3)$.
The size of a sequence is the total sum of its elements, so $(5,2,3,3)$ has size $13$.
We say that $v = (v_1,\dots,v_k)$ embeds into $u = (u_1,\dots,u_{k'})$ if
there exists an increasing function $f : [1,k] \to [1,k']$ such that for all $i \in [1,k]$,
we have $v_i \le u_{f(i)}$.
For example $(5,2,3,3)$ embeds into $(4,6,1,2,4,1,3)$ but not in $(3,2,5,3,3)$.
A sequence $u$ is $n$-universal if all sequences of size at most $n$ embed into $u$.

Let us define an $n$-universal sequence $u_n$, inductively on $n \in \N$.
We set $u_0 = ()$ (the empty sequence), $u_1 = (1)$, and 
$u_n$ is the concatenation of $u_{\lfloor n/2 \rfloor}$ with the singleton sequence $(n)$
followed by $u_{n - 1 - \lfloor n/2 \rfloor}$. 
Writing $+$ for concatenation, the definition reads $u_n = u_{\lfloor n/2 \rfloor} + (n) + u_{n - 1 - \lfloor n/2 \rfloor}$
Let us write the first sequences:
\[
u_2 = (1,2),\quad 
u_3 = (1,3,1),\quad
u_4 = (1,2,4,1),\quad
u_5 = (1,2,5,1,2),\quad
u_6 = (1,3,1,6,1,2), \dots
\]

\begin{lemma}
The sequence $u_n$ is $n$-universal and has size $O(n \log(n))$.
\end{lemma}
\begin{proof}
We proceed by induction on $n$. The case $n = 0$ is clear, let us assume that $n > 0$.
Let $v = (v_1,\dots,v_k)$ be a sequence of size $n$, we show that $v$ embeds into $u_n$.
There exists a unique $p \in [1,k]$ such that 
$(v_1,\dots,v_{p-1})$ has size smaller than or equal to $\lfloor n/2 \rfloor$
and $(v_1,\dots,v_p)$ has size larger than $\lfloor n/2 \rfloor$.
This implies that $(v_{p+1},\dots,v_k)$ has size at most $n - 1 - \lfloor n/2 \rfloor$.
By induction hypothesis $(v_1,\dots,v_{p-1})$ embeds into $u_{\lfloor n/2 \rfloor}$
and $(v_{p+1},\dots,v_k)$ embeds into $u_{n - 1 - \lfloor n/2 \rfloor}$,
so $v$ embeds into $u_n$.

The recurrence on size is $|u_n| = |u_{\lfloor n/2 \rfloor}| + n + |u_{n - 1 - \lfloor n/2 \rfloor}|$.
Solving it shows that $|u_n|$ is bounded by $O(n \log(n))$.
\end{proof}

We now use the universal sequence to improve on Lemma~\ref{lem:stronglyconnectedtogeneralnaive}.

\begin{lemma}
\label{lem:stronglyconnectedtogeneral}
Let $\Omega$ be a positionally determined prefix independent objective and $n \in \N$.
For each $k \in [1,n]$, let $\auto_k$ be a $(k,\Omega)$-separating automaton for strongly connected graphs.
Let us write $u_n = (x_1,\dots,x_k)$, 
then $\auto(u_n) = \langle \auto_{x_1},\dots,\auto_{x_k} \rangle$ is an $(n,\Omega)$-separating automaton.
\end{lemma}
\begin{proof}
We follow the same lines as for Lemma~\ref{lem:stronglyconnectedtogeneralnaive}, 
in particular the same argument implies that $\auto(u_n)$ satisfies $\Omega$.

Let $G$ be a graph satisfying $\Omega$, we show that $\IPath(G) \subseteq L(\auto(u_n))$.
We decompose $G$ into strongly connected components as before. 
Let us write $v = (|V(G_1)|,\dots,|V(G_p)|)$ the sequence of sizes of the components.
The sequence $v$ has size at most $n$, implying that $v$ embeds into $u_n$:
there exists an increasing function $f : [1,p] \to [1,|u_n|]$ such that
for all $i \in [1,p]$ we have $|V(G_i)| \le u_{f(i)}$.
It follows that for each $i \in [1,p]$, we have $\IPath(G_i) \subseteq L(\auto_{f(i)})$.

Let us consider a path $\pi$ in $G$, we show that $\pi$ is accepted by $\auto(u_n)$. 
The proof is by induction on the number of strongly connected components that $\pi$ traverses.

The base case is if $\pi$ traverses only one such component, say $G_i$. 
This implies that the run of $\pi$ on $\auto(u_n)$ either remains in the first $f(i) - 1$ copies, 
or reaches the $f(i)$ copy to read a suffix $\pi'$ of $\pi$. 
In the latter case, since $\IPath(G_i) \subseteq L(\auto_{f(i)})$ and $\pi'$ also remains in $G_i$, 
then $\auto_{f(i)}$ accepts $\pi'$. 
In both cases $\pi$ is accepted by $\auto(u_n)$.

Otherwise, let $G_i$ the first strongly connected component traversed by $\pi$.
The run of $\pi$ on $\auto(u_n)$ either remains in the first $f(i) - 1$ copies,
or reaches the $f(i)$ copy to read a suffix $\pi'$ of $\pi$.
Since $\IPath(G_i) \subseteq L(\auto_{f(i)})$, this run remains in the $f(i)$ copy as long as $\pi'$ remains in $G_i$.
This can either hold forever, in which case $\pi$ is accepted by $\auto(u_n)$,
or eventually the run jumps to the $f(i) + 1$ copy to read a suffix $\pi''$ of $\pi'$ (hence of $\pi$). 
In the latter case, since $\pi''$ traverses one less strongly connected component by induction hypothesis 
$\pi''$ is accepted using the copies from $f(i) + 1$ of $\auto(u_n)$. 
Thus $\pi$ is accepted by $\auto(u_n)$.
\end{proof}

To appreciate the improvement of Lemma~\ref{lem:stronglyconnectedtogeneral} over Lemma~\ref{lem:stronglyconnectedtogeneralnaive}, 
let us consider the case where the $(k,\Omega)$-separating automaton $\auto_k$ for strongly connected graphs has size $\alpha k$, where $\alpha$ does not depend on $k$ (see Theorem~\ref{thm:separatingautomatamp} for an example of such a case).
Then Lemma~\ref{lem:stronglyconnectedtogeneralnaive} yields an $(n,\Omega)$-separating automaton
of size $\alpha n^2$ while Lemma~\ref{lem:stronglyconnectedtogeneral} brings it down to $O(\alpha n \log(n))$.

\subsection{Separating automata for disjunctions of mean payoff objectives} 
\label{subsec:universalgraphsdisjunctionsmeanpayoff}
We state in the following theorem an upper bound on the construction of separating automata for disjunctions of mean payoff objectives.

\begin{theorem}
\label{thm:upperbounduniversalgraphdisjmeanpayoff}
Let $n,d,N \in \N$.
There exists an $(n,\bigvee_{i \in [1,d]} \MP{N}^i)$-separating automaton of size $O(n \log(n) \cdot d N)$.
\end{theorem}

As suggested by the previous subsection, we will start by constructing $(n,\bigvee_{i \in [1,d]} \MP{N}^i)$-separating automata
for strongly connected components.
The following result shows a decomposition property.

\begin{lemma}
Let $G$ be a strongly connected graph and $N \in \N$.
If $G$ satisfies $\bigvee_{i \in [1,d]} \MP{N}^i$ then there exists $i \in [1,d]$ such that $G$ satisfies $\MP{N}^i$.
\end{lemma}

\begin{proof}
We prove the contrapositive property: assume that for all $i \in [1,d]$ the graph does not satisfy $\MP{N}^i$,
implying that for each $i \in [1,d]$ there exists a negative cycle $C_i$ around some vertex $v_i$.
By iterating each cycle an increasing number of times, we construct a path in $G$ which does not satisfy $\bigvee_i \MP{N}^i$.

To make this statement formal, we use the following property:
for any $i \in [1,d]$, any finite path $\pi$ can be extended to a finite path $\pi \pi'$ with weight less than $-1$ on the $i$\textsuperscript{th} component.
This is achieved simply by first going to $v_i$ using strong connectedness, then iterating through cycle $C_i$ a sufficient number of times.

We then apply this process repeatedly and in a cyclic way over $i \in [1,d]$ to construct an infinite path
such that for each $i \in [1,d]$, infinitely many times the $i$-th component is less than $-1$.
This produces a path which does not satisfy $\bigvee_{i \in [1,d]} \MP{N}^i$, a contradiction.
\end{proof}

\begin{corollary}
\label{cor:universalgraphstronglyconnectedgraphsdisjmeanpayoff}
Let $n,N \in \N$ and $\auto$ be an $(n,\MP{N})$-separating automaton for strongly connected graphs.
We construct $d \cdot \auto$ the disjoint union of $d$ copies of $\auto$, where the $i$-th copy reads the $i$-th component.
Then $d \cdot \auto$ is an $(n,\bigvee_{i \in [1,d]} \MP{N}^i)$-separating automaton for strongly connected graphs. 
\end{corollary}

We can now prove Theorem~\ref{thm:upperbounduniversalgraphdisjmeanpayoff}.
Thanks to Theorem~\ref{thm:separatingautomatamp}, there exists an $(n,\MP{N})$-separating automaton of size $nN$.
Thanks to Corollary~\ref{cor:universalgraphstronglyconnectedgraphsdisjmeanpayoff}, this implies
an $(n,\bigvee_{i \in [1,d]} \MP{N}^i)$-separating automaton for strongly connected graphs of size $ndN$.
Now Lemma~\ref{lem:stronglyconnectedtogeneral} yields an $(n,\bigvee_i \MP{N}^i)$-separating automaton 
of size $O(n \log(n) \cdot d N)$.

\subsection{The complexity of solving disjunctions of mean payoff games using separating automata}
\label{subsec:complexitydisjunctionsmeanpayoff}

As in the previous case disjunction of mean payoff games were studied focussing on the objective of the opponent,
who has a conjunction of mean payoff objectives to satisfy~\cite{VelnerC0HRR15}.
Let us note here that there are actually two variants of the mean payoff objective: using infimum limit or supremum limit.
In many cases the two variants are equivalent: 
this is the case when considering mean payoff games or disjunctions of parity and mean payoff games.
However, this is not true anymore for disjunctions of mean payoff objectives, as explained in~\cite{VelnerC0HRR15}.
Our constructions and results apply using the infimum limit, and do not extend to the supremum limit.

The main result related to disjunction of mean payoff games in~\cite{VelnerC0HRR15} (Theorem~6) is that the problem is in $\NP \cap \coNP$ and can be solved in time $O(m \cdot n^2 \cdot d \cdot N)$.
(Note that since we consider the dual objectives, the infimum limit becomes a supremum limit in~\cite{VelnerC0HRR15}.)

Combining Theorem~\ref{thm:upperbounduniversalgraphdisjmeanpayoff}, Theorem~\ref{thm:separatingautomatamp},
and Theorem~\ref{thm:safetygames} yields the following result.

\begin{corollary}
Let $n,m,d,N \in \N$.
There exists an algorithm in the RAM model with word size $w = \log(n) + \log(N) + \log(d)$ for solving disjunctions of mean payoff games with weights in $[-N,N]$ 
with time complexity $O(m \cdot n \log(n) \cdot d \cdot N)$ and space complexity $O(n)$.
\end{corollary}

Note that for the choice of word size, the two tasks for manipulating separating automata, 
namely computing $\delta(q,w)$ and checking whether $q \le q'$ 
are indeed unitary operations as they manipulate numbers of order $nN$.